\documentclass{ifacconf}

\usepackage{graphics} 
\usepackage{epsfig} 
\usepackage{amsmath} 
\usepackage{amssymb}  
\usepackage{mathtools}
\usepackage{fancyhdr}
\usepackage{url}
\usepackage[dvipsnames]{xcolor}
\usepackage{natbib}        

\definecolor{antiquefuchsia}{rgb}{0.57, 0.36, 0.51}
\definecolor{DarkGreen}{rgb}{0.57, 0.8, 1}

\newtheorem{assumption}{Assumption}
\newtheorem{remark}{Remark}

\allowdisplaybreaks[4]

\begin{document}
\begin{frontmatter}

\title{Embodied Opinion Dynamics \\for
Safety-Critical Motion Control \\in Dynamic
Environments } 


\author[First]{Zhiqi Tang} 
\author[Second]{Yu Xing} 

\address[First]{Department of Electrical and Electronic Engineering,   
The University of Manchester, UK (e-mail: zhiqi.tang@manchester.ac.uk)}
\address[Second]{The Faculty of Computer Science, RWTH Aachen University, Aachen, Germany (e-mail: yu.xing@rwth-aachen.de)}

\begin{abstract}                
This paper proposes a novel adaptive control framework that embeds nonlinear opinion dynamics within the dynamical sensorimotor layers of an automated vehicle governed by second-order nonholonomic bicycle kinematics. The framework enables an ego vehicle to perform adaptive decision-making and achieve safe motion control under interaction uncertainty with non-cooperative neighboring agents. We consider a representative case study in which an ego vehicle autonomously attempts to merge into a lane occupied by human-driven or automated vehicles whose intentions are unknown. Within the proposed framework, the ego vehicle adaptively selects and executes merging versus non-merging behaviors in response to changing environmental conditions. Formal safety guarantees, as well as equilibrium and stability analyses of the closed-loop system, are provided. Numerical simulations further demonstrate the effectiveness of the proposed approach.
\end{abstract}


\end{frontmatter}

\section{Introduction}

Dynamical systems have provided valuable insights into how decision-making emerges in various biological contexts, including human and honeybee swarms \citep{gray2018multiagent}. These decision-making mechanisms are often value-based, such that the agent makes a decision when a specified variable crosses a static or dynamic threshold. Recent developments in nonlinear opinion dynamics \citep{bizyaeva2022nonlinear} have shown advantages in its flexible, tunable sensitivity, allowing the model to be either sensitive or robust to input.   However, these models generally neglect the agent’s control over its own physical dynamics. It remains an open question how to better systematize the advantages of these decision-making mechanisms in closed-loop control of physical systems. 


Inspired by the embodied decision-making of biological organisms whose sensory input, actions, and cognitive processes are
interconnected \citep{lepora2015embodied}, \cite{reverdy2018dynamical} and \cite{reverdy2021motivation} propose a motivation
dynamics framework (a dynamical systems framework) for a single robot to plan and execute recurrent coverage tasks. In the framework, value-based decision-making mechanisms are embedded directly within the dynamical sensorimotor layers of autonomous physical systems. Similarly, \cite{amorim2024threshold} propose a threshold-based decision-making framework for agents dynamically choosing between two spatial tasks that accounts for the agent's physical state. A key advantage of using a continuous motivation state in these approaches is the ability to encode a much richer set of behaviors than what is possible using a discrete switching variable. This richer behavioral repertoire has been also shown to facilitate deadlock resolution in robot navigation tasks. For instance, \cite{cathcart2023proactive} propose a proactive social-navigation framework in which a robot forms an “opinion’’ regarding how and in which direction to pass a human mover. Likewise, \cite{qi2025integrating} exploit nonlinear opinion dynamics within a safety controller to achieve blocking resolution without communication or predefined rules.

All of the works mentioned above consider only elementary robot dynamics, such as single-integrator models. Moreover, except for \cite{reverdy2018dynamical}, these approaches do not provide stability analysis and convergence guarantees for the physical state in the closed-loop system. Consequently, it remains an open problem to develop a general motivation-dynamics framework that provides formal performance guarantees for robots with more complex and realistic dynamics operating in dynamic environments.

Inspired by the approach proposed by \cite{reverdy2018dynamical}, this paper introduces a novel adaptive control framework that integrates nonlinear opinion dynamics within the dynamical sensorimotor layers of an automated vehicle governed by second-order kinematics of a nonholonomic bicycle model. The proposed framework is tailored for dynamic environments populated by non-cooperative neighboring agents, enabling the vehicle to perform adaptive decision-making and achieve safe motion control under interaction uncertainty.

Specifically, we consider a representative case study of a single ego vehicle autonomously merging into a lane occupied by a group of neighboring vehicles when the opportunity arises, while maintaining safety at all times. The primary challenge arises from that these neighboring vehicles may consist of both human-driven and autonomous vehicles that cannot coordinate with or communicate with the ego vehicle.  As a result, their intentions remain unknown, and they may or may not yield sufficient space for the ego vehicle to merge. Within the proposed framework, illustrated in Fig. \ref{fig:framework}, the opinion state of the ego vehicle represents its preference for executing a merging maneuver into the neighboring lane. This opinion is modeled as a dynamical system driven by the ego vehicle’s physical state relative to surrounding traffic. Conversely, the ego vehicle’s motion controller is informed by this opinion state, ensuring that the vehicle’s desired behavior emerges from the coupling between decision dynamics and physical dynamics.

\begin{figure}[!htb]
	\centering
	\includegraphics[scale = 0.45]{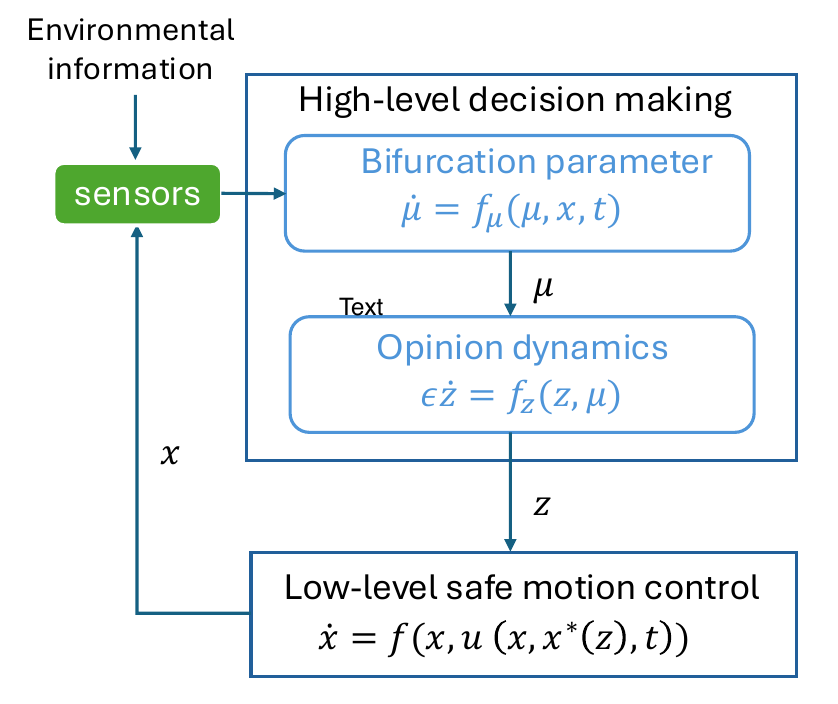}	
	\caption{The proposed control framework.}
		\label{fig:framework}
\end{figure}

In summary, the main contributions of the paper are as follows:
\begin{enumerate}
\item We propose a novel control framework that couples nonlinear opinion dynamics with the second-order nonholonomic bicycle dynamics of an ego vehicle, enabling the adaptive selection and execution of merging versus non-merging behaviors when interacting with non-cooperative neighboring agents.
\item We provide formal safety guarantees and equilibrium analysis by leveraging the constructive formulation of dissipative barrier feedback for safety-critical control.
\item We establish the stability of the overall closed-loop system by exploiting the time-scale separation between the opinion dynamics and the vehicle dynamics.
\end{enumerate}



The remainder of this paper is organized as follows. In Section~\ref{sec:problem}, we introduce the vehicle model and formulate the problem addressed in this work. Section~\ref{sec:design} presents the detailed design of the proposed framework along with the associated theoretical results. In Section~\ref{sec:simu}, we demonstrate the effectiveness of the method through numerical evaluations. Finally, Section~\ref{sec:conclusion} provides concluding remarks.


\section{ Vehicle model and Problem formulation}\label{sec:problem}
This section introduces the ego vehicle model and formally states the main problem addressed in this paper.


\subsection{Vehicle model}
\begin{figure}[t]
	\centering	\centerline{\includegraphics[trim={8cm 5cm 8cm 7cm},clip,width=1\linewidth]{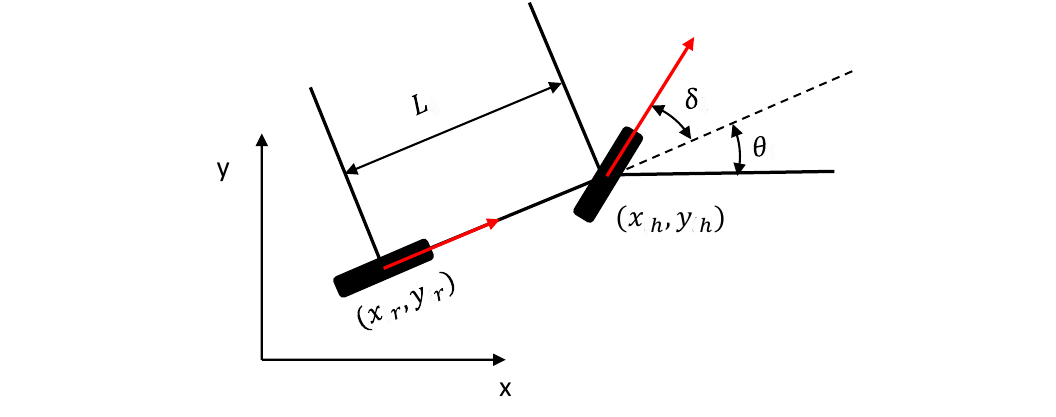}}
	\caption{Kinematic bicycle model for 2-dimensional vehicle motion.}
	\label{fig:bicycle}
\end{figure}
The ego vehicle is modeled using the second-order kinematics of a nonholonomic bicycle, as follows:

\begin{equation}\label{eq.bicycle_model}
\begin{bmatrix}
\dot{x}_{r}\\
\dot{y}_{r}\\
\dot{\theta}\\
\dot{v}_{r}\\
\dot{\delta}\\
\end{bmatrix}
=
\begin{bmatrix}
v_{r} \cos(\theta)\\
v_{r} \sin(\theta)\\
\frac{v_{r} \tan(\delta)}{L}\\
0\\
0\\
\end{bmatrix}
+\begin{bmatrix}
0 & 0\\
0 & 0\\
0 & 0\\
1 & 0\\
0 & 1\\
\end{bmatrix}
\begin{bmatrix}
a \\
\omega \\
\end{bmatrix}
\end{equation}
where $(x_{r}, y_{r})$ and $\theta$ indicate the rear axle center position and orientation of the ego vehicle in the common global frame, $v_{r}$ is the speed measured at the rear wheel, and finally $\delta$ , $L$ are the steering angle and wheelbase of vehicle, respectively, as shown in Fig. \ref{fig:bicycle}. The control inputs of the system are the longitudinal acceleration $a$ and the angular rate of the steering wheel $\omega$. 

In this paper, the focus is on the control design for the center of the front axle of the vehicle, modeled as a \textit{double integrator}:
\begin{equation}\label{eq:double integrator}
		\left\{
		\begin{aligned}
			\dot{p}&=v\\
			\dot{v}&=u
		\end{aligned}
		\right.
	\end{equation}
    where $p\in \mathbb R^2$ and $v\in \mathbb R^2$ are the position and velocity of of the front axle, and $u\in \mathbb R^2$ is the control input to be designed.
As shown in Fig. \ref{fig:bicycle}, the position of the front axle can be presented using the rear axle position
\begin{equation}
p=\begin{bmatrix}x_h\\y_h\end{bmatrix}=\begin{aligned}
 \begin{bmatrix}x_{r} + L\cos(\theta)\\
         y_{r} + L\sin(\theta)\end{bmatrix}.
\end{aligned}
    \end{equation}\label{eq.hand_pos}
     Take the second time derivatives of the above equation, the control input $(a,\omega)$ of the kinematic bicycle model~\eqref{eq.bicycle_model} can be transferred from the control design $u$  of the double integrator model similarly to 
     \cite{chen2024safe}, 
     as follows
\begin{equation}\label{eq.controlH2B}
\begin{aligned}
&\begin{bmatrix}
    a\\
    \omega\\
\end{bmatrix}
 = {\underbrace{\begin{bmatrix}
    \cos(\theta)-\sin(\theta)\tan(\delta) & -v_{r}\sin(\theta)\sec^2(\delta) \\
    \sin(\theta)+\cos(\theta)\tan(\delta) & v_{r}\cos(\theta)\sec^2(\delta) 
\end{bmatrix}}_A}^{-1}\\
& \left(u-
 \begin{bmatrix}
    -\frac{v_{r}^2}{L}\sin(\theta)\tan(\delta)-\frac{v_{r}^2}{L}\cos(\theta)\tan^2(\delta)\\
    \frac{v_{r}^2}{L}\cos(\theta)\tan(\delta)-\frac{v_{r}^2}{L}\sin(\theta)\tan^2(\delta)\\
\end{bmatrix}\right)
\end{aligned}
\end{equation}
This expression is valid as long as $\det A \neq 0$. Direct calculation gives $\det A = v_r \sec^2(\delta)$, so a solution exists as long as $\delta \neq \frac{\pi}{2}$ and $v_r \neq 0$. Physical constraints of the vehicle ensure $\delta < \frac{\pi}{2}$. To guarantee invertibility of $A$ when $|v_r| < \epsilon_v$ (with $\epsilon_v$ a small positive constant), $v_r$ in $A$ is replaced by $\text{sign}(v_r)\,\epsilon_v$.
\subsection{Problem formulation}
In this paper, we consider the problem of an ego vehicle that intends to autonomously merge into a lane occupied by a group of neighboring vehicles whenever an opportunity arises, while maintaining safety at all times. The challenge arises from the fact that the intentions of the neighboring vehicles are unknown. These vehicles may or may not yield space for the ego vehicle to merge.

Given the complexity of the problem, it is assumed that the ego vehicle has already identified the \textit{leader vehicle} to follow in the platoon of neighboring vehicles, in preparation for the merge. As illustrated in Fig.~\ref{fig:3agentB}, \textit{Vehicle 1} is designated as the leader, while \textit{Vehicle 2} is the vehicle immediately following \textit{Vehicle 1} in the same lane before the ego vehicle merges. Additionally, the following assumption is made regarding Vehicles 1 and 2:

\begin{assumption}\label{ass:human-driven}
For each neighboring vehicle $j \in \{1,2\}$, the position $p_j$, velocity $v_j$, and acceleration $u_j$ are bounded for all time. Each vehicle drives in the same lane and maintains a safe distance from other vehicles in that lane.
\end{assumption} 

\begin{figure}[!htb]
	\centering
	\includegraphics[scale = 0.4]{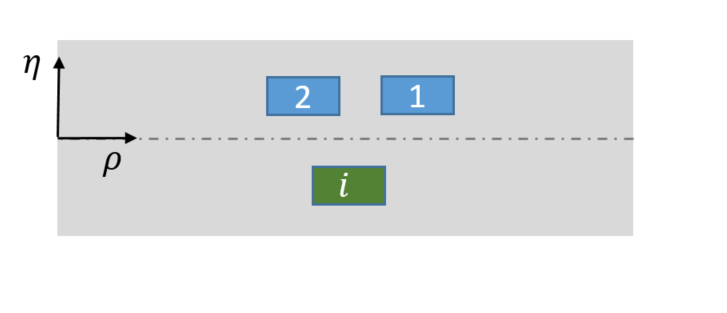}	
	\caption{Illustration of the ego vehicle (green) and neighboring vehicles (blue).}
		\label{fig:3agentB}
\end{figure}

We begin by describing the proposed framework in Fig.~\ref{fig:framework}, after which we present a formal problem formulation.
The first step involves designing an \textit{opinion} that represents the ego vehicle's preference for lane-changing. This opinion is denoted by $z \in \mathbb{R}$ with dynamics
\begin{equation}\label{eq:design-z}
\dot{z} = f_z(z, \mu),
\end{equation}
where $f_z(\cdot)$ is a nonlinear function to be designed, and $\mu \in \mathbb{R}$ is a bifurcation parameter that varies the equilibrium of $z$. Note that a change in the equilibrium corresponds to a continuous switch in the vehicle's opinion.

As the opinion $z$ should naturally be influenced by the physical state of the ego vehicle $x=[p^\top \ v^\top]^\top$ as well as the surrounding environment (i.e., the states of Vehicles 1 and 2), the dynamics of the bifurcation parameter are designed as
\begin{equation}\label{eq:design-mu}
\dot{\mu} = f_\mu(\mu, x, t),
\end{equation}
where $t$ denotes time.

To inform the vehicle dynamics with the opinion, the opinion is used to coordinate the desired state of the ego vehicle, resulting in closed-loop dynamics of the form
\begin{equation}\label{eq:design-x}
\dot{x} = f\big(x, u(x, x^*(z), t)\big),
\end{equation}
where $u(\cdot) \in \mathbb{R}^2$ is the safe motion controller to be designed.

Based on the above framework, the central problem addressed in this paper is the design of (i) suitable opinion dynamics \eqref{eq:design-z}, (ii) bifurcation parameter dynamics \eqref{eq:design-mu}, and~(iii) a safe motion controller $u(\cdot)$, such that the ego vehicle remains safe at all time while successfully executing a merging maneuver between Vehicles 1 and 2 when there is a chance.


\section{Proposed Framework Design}\label{sec:design}

\begin{figure}[t]
	\centering
	\includegraphics[scale = 0.3]{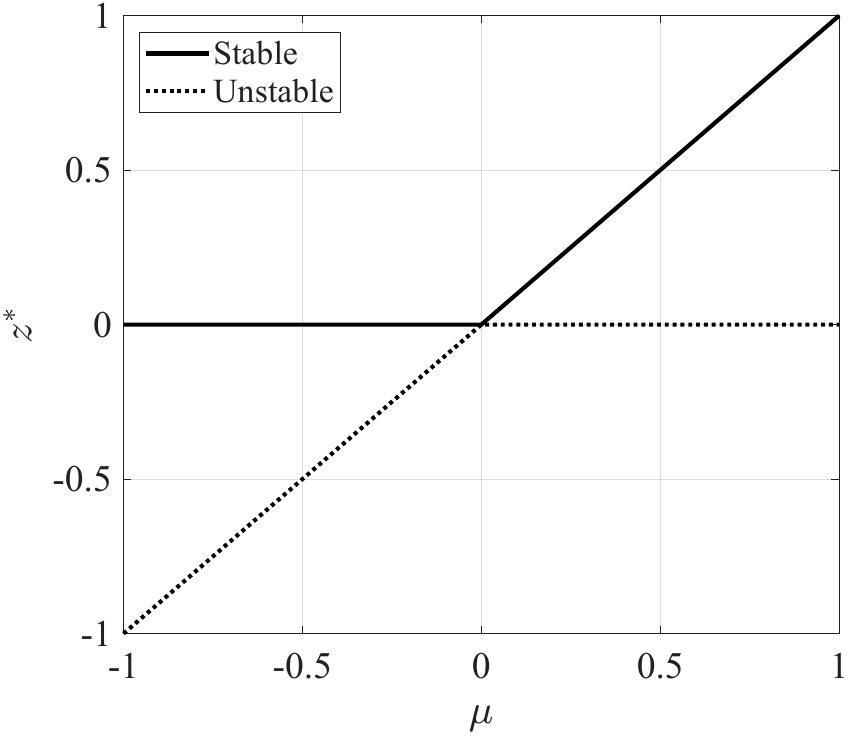}	
	\caption{Bifurcation diagram of the system~\eqref{eq:zi-2}. The Solid line represents stable equilibria, whereas the dotted line represents unstable ones.}
		\label{fig:bifurcation}
\end{figure} 

In this section, the detailed design of the proposed framework is presented.

Recall \eqref{eq:design-z}, let $z = 0$ represents the opinion to stay in the original lane, and $z > \epsilon_{1} > 0$ represents the opinion to change lanes. The opinion dynamics of $z$ is designed as\footnote{Note that this system with $\epsilon = 1$ corresponds to the normal form of a transcritical bifurcation~\cite[Section 4.2, Example 2]{perko2013differential}.} 
\begin{equation} \label{eq:zi-2}
\dot z = \frac{1}{\epsilon} (\mu z - z^2),
\end{equation}
where $\epsilon > 0$.

System \eqref{eq:zi-2} has two equilibrium points, $z^*_a = 0$ and $z^*_b = \mu$ when $\mu \neq 0$, and a unique equilibrium point $z^* = 0$ when $\mu = 0$ (see Fig.~\ref{fig:bifurcation}). The stability properties are summarized in the following lemma for completeness.

\begin{lem}\label{Z}
Consider system \eqref{eq:zi-2} with $\epsilon > 0$. If $z(0) \geq 0$, then $z(t) \geq 0$ for all $t \geq 0$. 

When $\mu \neq 0$, the system has two equilibrium points, $z^*_a = 0$ and $z^*_b = \mu$. 
\begin{itemize}
    \item If $\mu < 0$, $z^*_a$ is locally asymptotically stable, whereas $z^*_b$ is unstable. 
    \item If $\mu > 0$, $z^*_a$ is unstable, whereas $z^*_b$ is locally asymptotically stable. 
\end{itemize}
\end{lem}

\begin{pf}
For $z(0) \ge 0$, the right-hand side of \eqref{eq:zi-2}, $\dot z = \frac{1}{\epsilon} z (\mu - z)$, ensures $z(t) \ge 0$ for all $t \ge 0$.  
Equilibria satisfy $\dot z = 0$, giving $z^*_a = 0$ and $z^*_b = \mu$ when $\mu \neq 0$, and $z^* = 0$ when $\mu = 0$. Linearization yields $\frac{d}{dz}\dot z|_{z^*} = \frac{1}{\epsilon} (\mu - 2 z^*)$, which implies:  
\begin{itemize}
    \item $z^*_a = 0$ is stable if $\mu < 0$ and unstable if $\mu > 0$.
    \item $z^*_b = \mu$ is stable if $\mu > 0$ and unstable if $\mu < 0$.
\end{itemize}
\end{pf}

Next, the bifurcation parameter \eqref{eq:design-mu} is designed with the following dynamics:
\begin{equation}\label{eq:mu}
\dot \mu = -k_\mu\mu + \tanh\Big(-k\, \rho^\top g_1 \, \rho^\top g_2 \, (\bar d_{21} - 2r) \big(\frac{\dot{\bar d}_{21}}{\bar d_{21}} + \epsilon_1 \big) \Big),
\end{equation}
where $\rho \in \mathbb{S}^1$ denotes the constant longitudinal direction of the road, $g_j = \frac{p - p_j}{\|p - p_j\|}$ for $j \in \{1,2\}$ is the unit vector from the ego vehicle to neighboring Vehicle $j$,  
$\bar d_{21} = \|p_2 - p_1\| - r$ with $r>0$ representing a safe margin,  
$\dot{\bar d}_{21} = \bar g_{21}^\top v_{21}$ with $\bar g_{21} = \frac{p_2 - p_1}{\|p_2 - p_1\|}$ and $v_{21} = v_2 - v_1$, $k_\mu$ and $k$ are positive gains and $\epsilon_1$ is a positive scalar.

This design ensures that the opinion $z$ will deviate from its neutral state only if the following conditions are satisfied simultaneously:  
\begin{enumerate}
    \item The ego vehicle is in the correct relative position to merge, i.e., $\rho^\top g_1 \, \rho^\top g_2 < 0$.
    \item There is sufficient space for merging, i.e., $\bar d_{21} > 2r$.
    \item The divergent flow (ratio of relative velocity and relative distance between Vehicles 1 and 2) satisfies $\frac{\dot{\bar d}_{21}}{\bar d_{21}} > -\epsilon_1$.
\end{enumerate}






Finally, the method of dissipative barrier feedback \citep{tang2023collision} is employed for safe control design:
\begin{equation}\label{eq:a}
u = u^n + u^c,
\end{equation}
where the nominal tracking controller is defined as
\begin{equation}\label{eq:a_ni}
u^n = -k_d (e_1 - e_1^*(z)) - k_v \nu_1 + u_1,
\end{equation}
with 
\[
e_1 = p - p_1, \quad \nu_1 = v - v_1,
\]
and 
\begin{equation}\label{eq:e*}
e_1^* = \rho r^\rho + (1 - w(z)) \eta r^\eta
\end{equation}
denoting the desired relative position coordinated with the opinion $z$, where the weight is defined as $w(z) = \tanh(k_w z)$. Here, $\rho \in \mathbb{S}^1$ is the longitudinal direction of the road, $\eta \in \mathbb{S}^1$ is the lateral direction, $r^\rho$ and $r^\eta$ are constant offsets, and $u_1$ is the acceleration of the leader vehicle.  

To ensure safety, the dissipative barrier feedback is designed as
\begin{equation}
u^c = - \sum_{j \in \{1,2\}} k_o g_j \frac{\dot d_j}{d_j},
\end{equation}
where 
\[
d_j = \|p - p_j\| - r, \quad \dot d_j = g_j^\top \nu_j.
\]



The following lemma establishes formal safety guarantees, as well as equilibrium and stability analyses of the closed-loop system under the proposed framework.
\begin{lem} \label{lem:two-agents}
Consider two neighboring vehicles satisfying Assumption~\ref{ass:human-driven}, and an ego vehicle deciding whether to merge between them. The ego vehicle operates under the safe controller~\eqref{eq:a}, the opinion dynamics~\eqref{eq:zi-2}, and the bifurcation dynamics~\eqref{eq:mu}. Suppose the initial position $p(0)$ velocity $v(0)$ are bounded and safe according to Lemma~\ref{lem:boundness of OF} in the Appendix ($d_j(0) > 0$, $\frac{\dot d_j}{d_j}(0)$ bounded for $j \in \{1,2\}$). Assume also that the initial opinion satisfies $z(0) > 0$, and that the gains $k_p$, $k_v$, $k_o$, $k_w$, $k_\mu$, $k$,  and $\epsilon$ are positive and bounded. Then:
\begin{enumerate}
    \item The ego vehicle remains safe for all $t \ge 0$, i.e., $d_j(t) > 0$ and $\frac{\dot d_j}{d_j}(t)$ bounded for $j \in \{1,2\}$.
    \item There exists a sufficiently large gain $k_w > 0$ and a constant $\epsilon^* > 0$ such that, for all $0 < \epsilon < \epsilon^*$, the equilibrium points $(e_1, \nu_1) = (e_1^*(z^*), 0)$ are asymptotically stable, provided that $v_{21}(t)$ converges to zero.
\end{enumerate}
\end{lem}
\begin{remark}
Note that $v_{21}$ is not required to converge to zero, as assumed in Assumption~\ref{ass:human-driven}. When $v_{21}$ does not converge to zero, the closed-loop system~\eqref{eq:cascaded_i1} is input-to-state stable with respect to $v_{21}$.
\end{remark}
\begin{pf}
Proof of item 1):\\
Recall that $\dot d_{j}=g_{j}^\top \nu_{j}, j\in \{1,2\}$, and hence one verifies that:
\begin{equation}\label{dotdij}
\ddot d_{j}=-k_o\frac{\dot d_{j}}{d_{j}}-k_v\dot d_{j}-k_og_{j}^\top g_{l}\frac{\dot d_{l}}{d_{l}}+\alpha_{j}, l\in \{1,2\}, l\ne j,
\end{equation}
where \begin{align*}
\alpha_{1}&=-g_{1}^\top k_p(e_{1}-e_{1}^*(z))+\frac{\|\pi_{g_{1}}\nu_{1}\|^2}{d_{1}+r}, \\
\alpha_{2}&=-g_{2}^\top(k_p(e_{1}-e_{1}^*(z))+u_2-u_1+v_{21})+\frac{\|\pi_{g_{2}}\nu_{2}\|^2}{d_{2}+r}
\end{align*}
 with $v_{21}=v_2-v_1$ and $\pi_{y}=I-yy^\top$ the projection operator for $S^1$. 
We prove that $d_j$ remains positive using contradiction, in two cases.

i) 
Assume that one of $d_{1}$ or $d_{2}$ approaches zero at a finite time $T$, while the other remains positive on $[0,T]$.

Integrating \eqref{dotdij} from $0$ to $T$ gives:
\begin{equation}\label{eq:int_dij}
\begin{aligned}
&k_o(\ln d_{j}(T)-\ln d_{j}(0))\\
&=-\dot d_{j}(T)+\dot d_{j}(0)-k_v(d_{j}(T)-d_{j}(0))\\
&+\int_0^T(\alpha_{j}-k_og_{j}^\top g_{l}\frac{\dot d_l}{d_l})d\tau
\end{aligned}
\end{equation}

The left-hand side of the equation tends to negative infinity, however the right-hand side of the equation is either bounded or tends to positive infinity. This is because as $d_{j}$ approaches zero, $\dot d_{j}$ is either bounded or negative infinity; $e_{1}$,  $e_{1}^*(z)$, and $d_l \ (l\in \{1,2\}, l\ne j$) remain positive, hence  $\frac{\dot d_l}{d_l}$ is also bounded; $u_2$ and $v_{12}$ are bounded by Assumption \ref{ass:human-driven}, and $\alpha_{j}$ is either bounded or tends to positive inifinity. This yields a contradiction.

ii) Assume now that $d_{1}$ and $d_{2}$ approach zero simultaneously. 
Summing $\ddot d_{1}$ and $\ddot d_{2}$ yields:
\begin{equation}\label{eq:sumd}
\ddot d_{1}+\ddot d_{2}=-k_o(1+g_{1}^\top g_{2})(\frac{\dot d_{1}}{d_{1}}+\frac{\dot d_{2}}{d_{2}})-k_v(\dot d_{1}+\dot d_{2})+\alpha_{1}+\alpha_{2}
\end{equation}
Using the fact that $|g_{1}^\top g_{2}|\le 1$ and Under Assumption~\ref{ass:human-driven}, if  
$-1 \le g_{1}^\top g_{2} \le -1+\epsilon_2$ (for some small $\epsilon_2>0$),  
the ego vehicle has already merged, so $d_{1}$ and $d_{2}$ cannot both go to zero simultaneously. Hence, we only consider $g_{1}^\top g_{2} > -1+\epsilon_2$ on $[0,T]$.
 
 Integrating \eqref{eq:sumd}, it yields:
\begin{equation}
\begin{aligned}
&\beta\sum_{j\in \{1,2\}}k_o(\ln d_{j}(T)-\ln d_{j}(0))\\
&=\sum_{j\in\{1,2\}}\big(-\dot d_{j}(T)+\dot d_{j}(0)-k_vd_{j}(T)\\
&+k_vd_{j}(0)+\int_0^T\alpha_{j}d\tau\big)
\end{aligned}
\end{equation}
where $\beta \in [\min(1-g_{1}^\top g_{2}), \max(1-g_{1}^\top g_{2})]$ is a positive scalar, such that :

	\scalebox{0.9}{
		$\int^T_0 (1+g_{1}^\top g_{2})(\frac{\dot d_{1}} {d_{1}}+\frac{\dot d_{2}} {d_{2}}) d\tau =\beta\sum_{j\in \{1,2\}}(\ln d_{j}(T)-\ln d_{j}(0)) $.}
        
Using the same contradiction argument as for \eqref{eq:int_dij}, we conclude that $d_{1}$ and $d_{2}$ cannot approach zero simultaneously.      

Proof of item 2):\\
Consider the dynamics of $(e_{1},\nu_{1})$:
\begin{equation}\label{eq:cascaded_i1}
\left\{
\begin{aligned}
\dot e_{1}=&\nu_{1}\\
\dot \nu_{1}=&-k_p(e_{1}-e_{1}^*(z))-k_v\nu_{1}-k_o\sum_{j\in\{1,2\}} \frac{g_{j}g_{j}^\top \nu_{j}}{d_{j}}.
\end{aligned}
\right.
\end{equation}
Using Lemma~\ref{Z}, substitute $z^{*}\in\{0,\mu\}$ into \eqref{eq:cascaded_i1}, it yields:
\begin{equation}\label{eq:cascaded_z*}
\left\{
\begin{aligned}
\dot e_{1}=&\nu_{1}\\
\dot \nu_{1}=&-k_p(e_{1}-e_{1}^*(z^*))-k_v\nu_{1}-k_o\sum_{j\in\{1,2\}} \frac{g_{j}g_{j}^\top \nu_{j}}{d_{j}}
\end{aligned}
\right.
\end{equation}
Since $\nu_{2}=\nu_{1}-v_{21}$, the system \eqref{eq:cascaded_z*} is a cascaded system perturbed by $v_{21}$.  The unforced subsystem is:
\begin{equation}\label{eq:unforced_i1}
\left\{
\begin{aligned}
\dot e_{1}=&\nu_{1}\\
\dot \nu_{1}=&-k_p(e_{1}-e_{1}^*(z^*))-k_v\nu_{1}-k_o\sum_{j\in\{1,2\}} \frac{g_{j}g_{j}^\top \nu_{1}}{d_{j}}
\end{aligned}
\right.
\end{equation}

Consider the following Lyapunov function candidate
\begin{equation}
\mathcal L=\frac{k_p}2 \|e_{1}-e_{1}^*(z^*)\|^2+\frac 1 2 \|\nu_{1}\|^2
\end{equation}
Recall \eqref{eq:e*} and using the fact that $k_w$ is large enough, one has $w(z^*_{i_a})=0$ and $w(z^*_{i_b})\approx 1$ and hence one can consider $e_{i1}^*(z^*_{i_a})=\rho r^\rho +\eta r^\eta$ and $e_{i1}^*(z^*_{i_b})=\rho r^\rho$ as constant. In this setting, one verifies that:
\begin{equation}
\dot {\mathcal L}=-k_v\|\nu_{1}\|^2-k_o\sum_{j\in \{1,2\}}|g_{j}^\top \nu_{1}|^2/d_{j}\le 0
\end{equation}
which implies boundedness of $(e_{1},\nu_{1})$, as long as $d_{j}>0$. Boundedness of $\dot\nu_{1}$, $\dot g_{j}$, and $\dot d_{j}$ ensures that $\ddot{\mathcal L}$ is bounded. 

 Using Barbalet's lemma, one concludes that the equilibrium points $(e_{1},\nu_{1})=(e_{1}(z^*), 0)$ of the unforced system \eqref{eq:unforced_i1} are asymptotically stable. If $v_{21}\to 0$, the same conclusion holds for the cascaded system \eqref{eq:cascaded_z*}.

Finally, define $\xi=[e_{1}^\top\ \nu_{1}^\top]^\top$ and write the singular perturbation form
\begin{equation}
\begin{aligned}
\dot \xi&=f_e(t,\xi,z,\epsilon)\\
\epsilon\dot z&=f_z(t,z,\mu,\epsilon)\\
\dot \mu&=f_\mu(t,x_e,\epsilon).
\end{aligned}
\end{equation}

Using Lemma~\ref{Z}, the stability of \eqref{eq:cascaded_z*}, and the boundedness of
\(f_e\), \(f_z\), an argument analogous
to \cite[Section 7.5, Theorem 5.1]{kokotovic1986singular} implies that, for sufficiently small
\(\epsilon>0\), the equilibrium
\[
(e_1,\nu_1)
=
(e_1^*(z^*),0)
\]
of the full system \eqref{eq:cascaded_i1} is locally asymptotically stable on
each stable opinion branch, provided \(v_{21}\to0\). This conclusion is applied branch-wise away from the bifurcation point
\(\mu=0\), where the stable equilibrium of the fast opinion dynamics is locally
asymptotically stable.



\end{pf}




\section{Numerical results}\label{sec:simu}
This section presents a numerical simulation of the proposed framework. To demonstrate the advantages of the method, we consider an aggressive driving behavior for Vehicle~2, in which it intermittently changes its intention between yielding and not yielding to the ego vehicle. Operationally, this corresponds to Vehicle~2 switching between deceleration and acceleration. The gains and parameters are chosen as $k_p=0.7, k_v=2, k_o=1, k_w=40, k_\mu=5, k=20, \epsilon=0.05, r=0.6, \epsilon_1=0.5$. The animation illustrating the interaction among Vehicle 1, Vehicle 2, and the ego vehicle under the proposed control framework is available at: \textcolor{blue}{https://bit.ly/3Kh6tHf}.
Figure~\ref{fig:z} shows the evolution of the opinion state 
$z$, together with the ego vehicle’s lateral position along the road. From the video and Fig.~\ref{fig:z}, we observe that when $t<10$, the opinion state transitions continuously from a small positive value toward zero in response to the aggressive behavior of Vehicle 2, which alternates between yielding and not yielding to the ego vehicle. After $t=20$, the ego vehicle merges between Vehicle 1 and 2. Owing to the time-scale separation in the design, the physical state evolves smoothly and does not exhibit oscillatory behavior. Finally, the comparison between Fig.~\ref{fig:d} and Fig.~\ref{fig:d-collide} demonstrates the effectiveness of the proposed framework in achieving reactive collision avoidance in a dynamic environment.

\section{Conclusions}\label{sec:conclusion}
This work presented an adaptive embodied decision-making framework that integrates nonlinear opinion dynamics with the safety-critical control of a nonholonomic autonomous vehicle. By coupling decision states with physical dynamics, the proposed method enables an ego vehicle to react to interaction uncertainty and adjust its behavior accordingly while maintaining formal safety guarantees. Analytical results established equilibrium properties and closed-loop stability, and numerical simulations demonstrated effective behavior adaptation and safe motion execution in dynamic non-cooperative environments. Future work will extend this framework to multi-vehicle coordination and more complex interaction scenarios.
\begin{figure}[t]
	\centering
	\includegraphics[scale = 0.4]{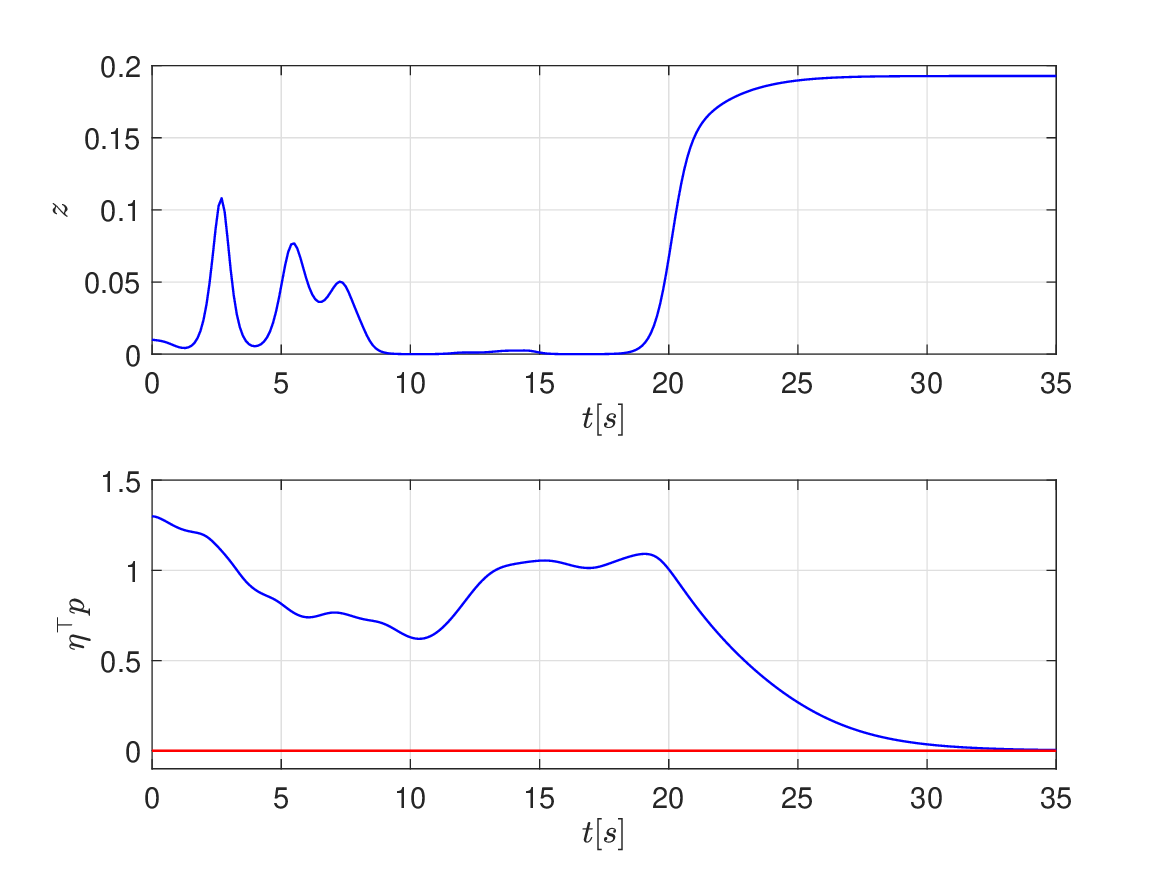}	
	\caption{Evolution of opinion state $z$ and the position of ego vehicle projected on the lateral direction of the road $\eta^\top p$. The red line is the lateral position of the lane containing Vehicle 1 and 2.}
		\label{fig:z}
\end{figure} 

\begin{figure}[t]
	\centering
	\includegraphics[scale = 0.4]{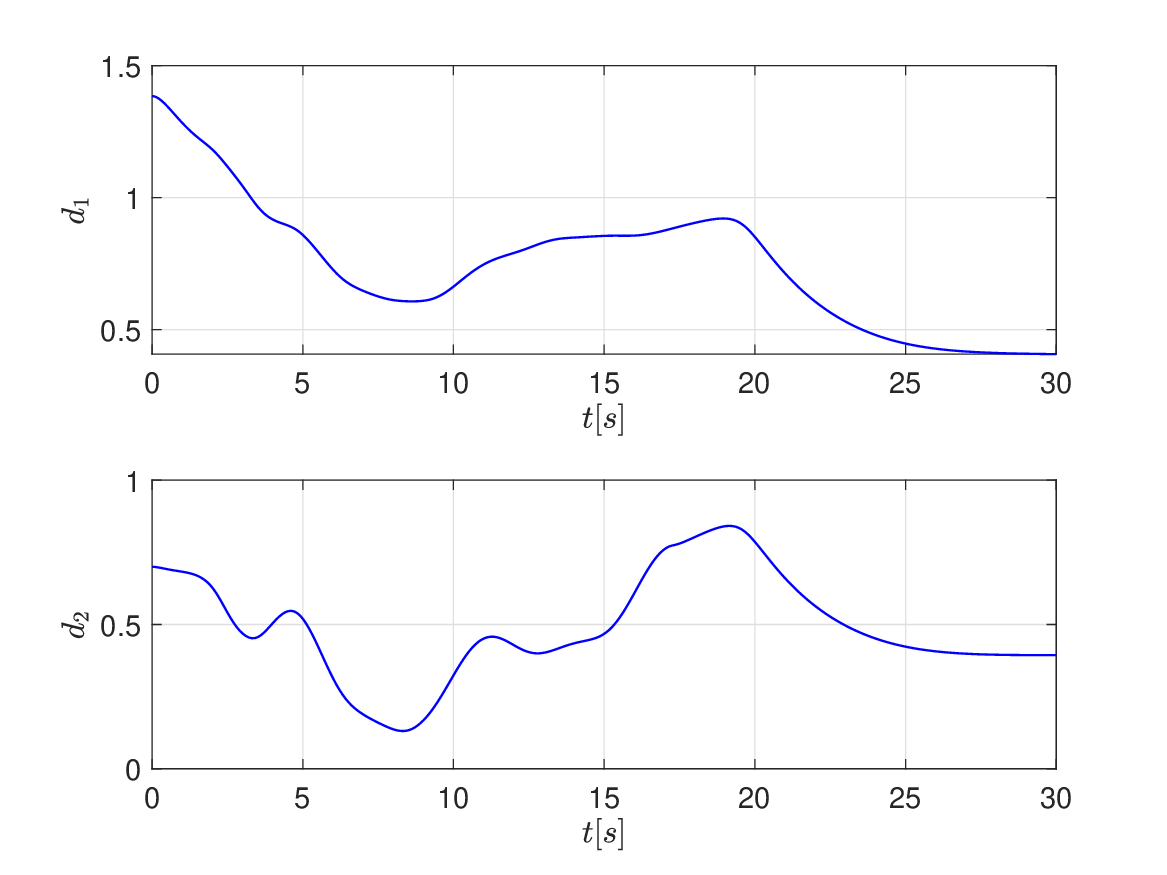}	
	\caption{Evolution of the safety distance $d_1$ and $d_2$ under the proposed framework with collision avoidance term $u_c$. $d_1$ and $d_2$ remain positive.}
		\label{fig:d}
\end{figure} 

\begin{figure}[t]
	\centering
	\includegraphics[scale = 0.4]{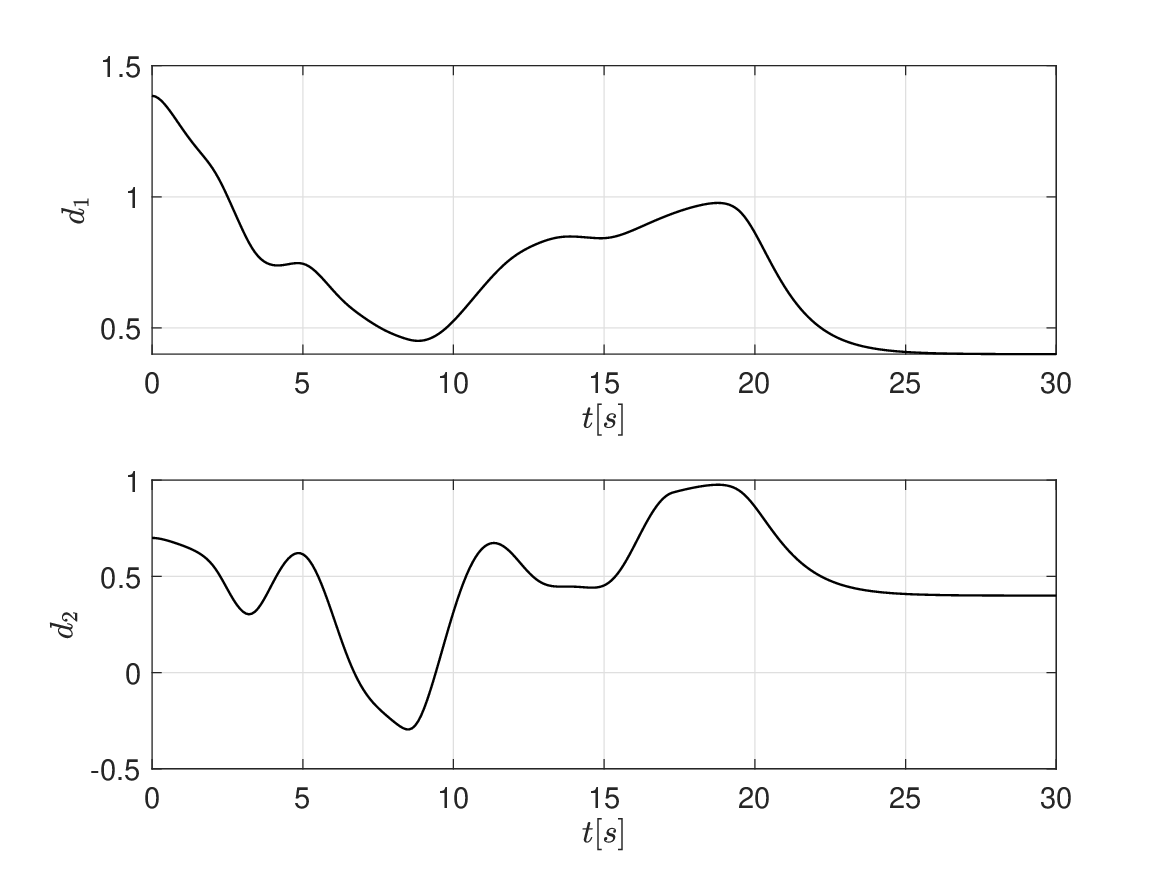}	
	\caption{Evolution of the safety distance $d_1$ and $d_2$ under the proposed framework without collision avoidance term $u_c$ ($k_o=0$). The safety constraint is violated ($d_2$ is negative) at some point between $t=5$ and $t=10$. }
		\label{fig:d-collide}
\end{figure} 


\section*{DECLARATION OF GENERATIVE AI AND AI-ASSISTED TECHNOLOGIES IN THE WRITING PROCESS}
During the preparation of this work, the authors used ChatGPT in order to improve the language of some paragraphs. After using this tool/service, the authors reviewed and edited the content as needed and take full responsibility for the content of the publication.

\bibliography{bibliography}             

\appendix
\section*{Appendix}

\begin{lem} \label{lem:boundness of OF}
Given the dynamics 
\begin{equation} \label{ddot_d}
		\ddot d=-k_o\frac{\dot d }{d} -\alpha(t) 
	\end{equation}
			with $k_o$ a positive gain and $\alpha(t)$ a continuous and bounded function. Then for any initial condition satisfying $d(0)>0$ and $\phi(0)=\frac{\dot d(0)} {d(0)}$ bounded, the following assertions hold:
			\begin{enumerate}
				\item $d$ remains positive, $\forall t\ge 0$.
				\item $d$ converges to zero as $t\to \infty$ if and only if $\lim_{t \to \infty} \int^\top_0 \alpha(\tau) d\tau \to +\infty$. 
				\item If $d$ converges to zero, then $\dot{d}$ is bounded and converges to zero, and $\phi(t)$ remains bounded, $\forall t\ge 0$. Furthermore, if $\alpha(t)$ converges to a positive constant $\alpha^0>\epsilon >0$, then $\frac{\dot d}{d}\to -\frac{\alpha^0}{k_o}$ and hence $\ddot d$ converges to zero.
			\end{enumerate}
	\end{lem}

This lemma shows that as long as the initial distance $d(0)$ is positive and $\phi(0)$ is bounded, then $d(t)$ will never cross zero, and $\phi(t)$ remains bounded. The proof of this Lemma can be found in \cite{tang2023collision}.

\end{document}